\documentstyle[12pt,fleqn]{article}

\font\egtgot =eufm8 \font\sevgot =eufm7

\font\twlmsb =msbm10 scaled \magstep1 \font\egtmsb =msbm8
\font\sevmsb =msbm7

\newfam\gotfam
\def\pgot{\fam\gotfam\twlgot}
\textfont\gotfam\twlgot \scriptfont\gotfam\egtgot
\scriptscriptfont\gotfam\sevgot
\def\got{\protect\pgot}

\newfam\msbfam
\textfont\msbfam\twlmsb \scriptfont\msbfam\egtmsb
\scriptscriptfont\msbfam\sevmsb
\def\Bbb{\protect\pBbb}
\def\pBbb{\relax\ifmmode\expandafter\Bb\else\typeout{You cann't use
Bbb in text mode}\fi}
\def\Bb #1{{\fam\msbfam\relax#1}}

\def\op#1{\mathop{{\it\fam0} #1}\limits}

\newcommand{\id}{{\rm Id\,}}

\newcommand{\di}{{\rm dim\,}}

\newcommand{\hm}{{\rm Hom\,}}
\newcommand{\dif}{{\rm Diff\,}}

\newcommand{\beq}{\begin{equation}}
\newcommand{\eeq}{\end{equation}}
\newcommand{\ben}{\begin{eqnarray}}
\newcommand{\een}{\end{eqnarray}}
\newcommand{\be}{\begin{eqnarray*}}
\newcommand{\ee}{\end{eqnarray*}}
\newcommand{\bea}{\begin{eqalph}}
\newcommand{\eea}{\end{eqalph}}

\newcommand{\cO}{{\cal O}}
\newcommand{\cA}{{\cal A}}

\newcommand{\gd}{{\got d}}

\newcommand{\gA}{{\got A}}

\newcommand{\cT}{{\cal T}}

\newcommand{\cV}{{\cal V}}

\newcommand{\cZ}{{\cal Z}}

\newcommand{\cK}{{\cal K}}

\newcommand{\dl}{\delta}
\newcommand{\la}{\lambda}
\newcommand{\La}{\Lambda}
\newcommand{\f}{\phi}

\newcommand{\Om}{\Omega}
\newcommand{\m}{\mu}

\newcommand{\g}{\gamma}

\newcommand{\e}{\epsilon}
\newcommand{\ve}{\varepsilon}
\newcommand{\th}{\theta}

\newcommand{\bb}{{\bf 1}}

\newcommand{\w}{\wedge}

\newcommand{\wh}{\widehat}

\newcommand{\dr}{\partial}

\newcommand{\ar}{\op\longrightarrow}

\newcommand{\ot}{\otimes}

\newenvironment{eqalph}{\stepcounter{equation}
\setcounter{equationa}{\value{equation}} \setcounter{equation}{0}

\begin{eqnarray}}{\end{eqnarray}\setcounter{equation}{\value{equationa}}}

\newcounter{example}
\newcounter{remark}
\newcounter{theorem}
\newcounter{proposition}
\newcounter{lemma}
\newcounter{corollary}
\newcounter{definition}

\def\theremark{\arabic{remark}}

\def\thetheorem{\arabic{theorem}}

\def\thedefinition{\arabic{definition}}

\newenvironment{proof}{{\bf Proof.}}{
\hfill $\Box$ \medskip }

\newenvironment{ex}{\refstepcounter{remark} \medskip {\bf Example
\theremark.} }{ \medskip }

\newenvironment{lem}{\refstepcounter{theorem} \medskip{\bf Lemma
\thetheorem.}\it }{\medskip }

\newenvironment{defi}{\refstepcounter{definition} \medskip{\bf
Definition \thedefinition.} \it}{\medskip }

\newcommand{\mar}[1]{}

\begin{document}
\hbox{}

{\parindent=0pt

{\large \bf Differential operators on Lie and graded Lie algebras}
\bigskip

{\sc G. Sardanashvily}

\medskip

{\it Department of Theoretical Physics, Moscow State University,
117234 Moscow, Russia}

E-mail: gennadi.sardanashily@unicam.it

\bigskip

\bigskip

\begin{small}

 {\bf Abstract.} Theory of differential operators on associative algebras
is not extended to the non-associative ones in a straightforward
way. We consider differential operators on Lie algebras. A key
point is that multiplication in a Lie algebra is its derivation.
Higher order differential operators on a Lie algebra are defined
as composition of the first order ones. The Chevalley--Eilenberg
differential calculus over a Lie algebra is defined. Examples of
finite-dimensional Lie algebras, Poisson algebras, algebras of
vector fields, and algebras of canonical commutation relations are
considered. Differential operators on graded Lie algebras are
defined just as on the Lie ones.
\medskip

\end{small}

}

\section{Introduction}

There is the conventional notion of (linear) differential
operators on commutative rings \cite{book05,grot,kras,sard9}. This
notion is straightforwardly extended to graded commutative rings
in supergeometry \cite{book05,sard9}, but its generalization to
non-commutative rings in non-commutative geometry is not unique
\cite{bor97,dublmp,dub01,book05,lunts,sard9}. The underlying
problem is that, being first order differential operators,
derivations of a non-commutative ring $\cA$ fail to form an
$\cA$-module. At the same time, all the existed definitions of
differential operators on rings (i.e., unital associative algebras
with $\bb\neq 0$) follow that on commutative rings (see
Definitions \ref{ws131} and \ref{s93} below).

Let $\cK$ be a commutative ring and $\cA$ a commutative
$\cK$-ring. Note that any associative $\cK$-algebra $\cA$ can be
extended to a ring by the adjunction of the unit element $\bb$ to
$\cA$. Let $\Phi\in\hm_\cK(\cA,\cA)$ be an endomorphism of a
$\cK$-module $\cA$. Given an element $a\in \cA$, let us define an
endomorphism
\mar{spr172}\beq
(\dl_a\Phi)(c)= a\Phi(c) -\Phi(ac), \qquad \forall c\in\cA,
\label{spr172}
\eeq
of a $\cK$-module $\cA$ and its endomorphism
$\dl_{a_0}\circ\cdots\circ\dl_{a_k}\Phi$ for any tuple of elements
$a_0,\ldots,a_k$ of $\cA$.

\begin{defi} \label{ws131} \mar{ws131}
An element $\Delta\in\hm_\cK(\cA,\cA)$ is called a $k$-order
differential operator on a commutative $\cK$-ring $\cA$ if
\mar{s5a}\beq
\dl_{a_0}\circ\cdots\circ\dl_{a_k}\Delta=0 \label{s5a}
\eeq
for any tuple of $k+1$ elements $a_0,\ldots,a_k$ of $\cA$.
\end{defi}

In particular, a zero order differential operator $\Phi\in
\dif_0(A)$ on $A$ obeys the condition $\dl_a\Phi=0$, $\forall a\in
A$, i.e., it is an $\cA$-module endomorphism of $\cA$. Therefore,
all zero order differential operators on $\cA$ are multiplications
in $\cA$. A first order differential operator $\Delta$ on a
commutative ring $A$ satisfies the condition
\mar{ws106}\beq
\dl_b\circ\dl_a \Delta= 0, \qquad \forall a,b\in\cA. \label{ws106}
\eeq
For instance, any derivation of a commutative ring $\cA$ obeys the
relation (\ref{ws106}) and, consequently, it is a first order
differential operator. Moreover, any first order differential
operator $\Delta\in \dif_1(A)$ on $\cA$ takes the form
\mar{s95'}\beq
\Delta(b)=\dr b +ab, \qquad \forall b\in\cA, \label{s95'}
\eeq
where $\dr$ is a derivation of $\cA$ and $a\in \cA$.

Let $P$ be a bimodule over a commutative ring $\cA$. Let
$\Phi\in\hm_\cK(P,P)$ be an endomorphism of a $\cK$-module $P$.
Given an element $a\in \cA$, let us define an endomorphism
\be
(\dl_a\Phi)(p)= a\Phi(p) -\Phi(ap), \qquad \forall p\in P,
\ee
of a $\cK$-module $P$ and its endomorphism
$\dl_{a_0}\circ\cdots\circ\dl_{a_k}\Phi$ for any tuple of elements
$a_0,\ldots,a_k$ of $\cA$.

\begin{defi} \label{s93} \mar{s93}
An element $\Delta\in\hm_\cK(P,P)$ is called a $k$-order
$P$-valued differential operator on an $\cA$-module $P$ if
\be
\dl_{a_0}\circ\cdots\circ\dl_{a_k}\Delta=0
\ee
for any tuple of $k+1$ elements $a_0,\ldots,a_k$ of $\cA$.
\end{defi}

In particular, a zero order differential operator $\Phi\in
\dif_0(P)$ on $P$ is an $\cA$-module endomorphism of $P$. A first
order differential operator $\Delta\in\dif_1(P)$ on $P$ is a
$\cK$-module endomorphism of $P$ which obeys the condition
\mar{s95}\beq
\Delta(ap)=(\dr a)p +\Delta(p), \qquad \forall a\in \cA.
\label{s95}
\eeq

It should be emphasized that derivations of rings also are first
order differential equations in supergeometry and non-commutative
geometry. Therefore, a definition of differential operators on Lie
algebras must treat derivations of Lie algebras as first order
differential operators, too. However, Definition \ref{ws131} fails
to be generalized to non-associative algebras because their
derivations do not satisfy the condition (\ref{ws106}). Therefore,
a comprehensive definition of differential operators on
non-associative algebras fails to be formulated.

This work addresses differential operators on Lie and $\Bbb
Z_2$-graded Lie algebras and modules over these algebras. A key
point is that multiplications in these algebras are their
derivations, i.e., first order differential operators. We restrict
our consideration of differential operators on Lie and graded Lie
algebras to compositions of first order differential operators
(see Definition \ref{s20} below). However, it may happen that
there exist other operators which can be treated as higher order
differential operators (see Examples \ref{s16} and \ref{s41}).

There are different variants of differential calculus over a ring
\cite{conn,land}. Following a notion of the Chevalley--Eilenberg
differential calculus over a ring \cite{dub01,book05,sard9}, we
define the Chevalley--Eilenberg differential calculus on a Lie
algebra (see Definition \ref{s50} below). If a Lie algebra has a
zero center and all its derivations are inner, this differential
calculus coincides with the well-known Chevalley--Eilenberg
complex of this algebra \cite{fuks}. If a Lie algebra is
finite-dimensional, this complex describes the matrix geometry
(see Example \ref{s55} below). This also is the case of a
finite-dimensional graded Lie algebra in Example \ref{s120}.

Physically relevant examples of finite-dimensional Lie and graded
Lie algebras, Poisson algebras, algebras of vector and graded
vector fields, algebras of canonical commutation and
anticommutation relations are considered.

\section{Derivations of Lie algebras}

Let $\cK$ be a commutative ring and $A$ a Lie algebra over $\cK$.
Let the symbol $\cdot$ stand for the multiplication in this
algebra such that its standard properties read
\mar{ss1}\beq
a\cdot b=-b\cdot a,\qquad a\cdot (b\cdot c) + b\cdot (c\cdot a) +
c\cdot (a\cdot b) =0, \qquad \forall a,b,c,\in A. \label{ss1}
\eeq
One also regards $A$ as an $(A-A)$-module where left and right
$A$-module structures are related by the condition
\be
ap=a\cdot p=-p\cdot a=-pa, \qquad \forall p,a\in A.
\ee

\begin{defi} A derivation of a Lie algebra $A$ is defined as its $\cK$-module
automorphism $\dr$ which obeys the Leibniz rule
\be
\dr(a\cdot b)=\dr a\cdot b +a\cdot \dr b, \qquad \forall a,b\in A.
\ee
\end{defi}

For instance, any left multiplication
\mar{s1}\beq
c: A\ni a\to c\cdot a\in A \label{s1}
\eeq
in $A$ is its inner derivation by virtue of the Jacobi identity
(\ref{ss1}):
\be
c\cdot (a\cdot b)= (c\cdot a)\cdot b + a\cdot (c\cdot b).
\ee
Certainly, any right multiplication also is well. Let $\dr_c$
further denote the inner derivation (\ref{s1}).

Derivations of a Lie algebra $A$ constitute a Lie algebra $\gd A$
over a ring $\cK$ with respect to the Lie bracket
\be
[\dr,\dr']=\dr\circ\dr'-\dr'\circ\dr.
\ee
For instance, the Lie bracket of inner derivations reads
\be
[\dr_a,\dr_b]=\dr_{a\cdot b}.
\ee
Thus, there exists a Lie algebra homomorphism $A\ni a\to
\dr_a\in\gd A$ whose kernel is the center $\cZ[A]$ of $A$
consisting of devisors of zero of $A$ (i.e., elements $a\in A$
such that $a\cdot b=0$ for all $b\in A$).

Let us consider a few physically relevant examples.

\begin{ex} \label{s3} \mar{s3}
If $A$ is an $N$-dimensional Lie algebra, the Lie algebra $\gd A$
of its derivations is of dimension $\di\gd A\leq N^2$. For
instance, $\gd gl(k,\Bbb R)=sl(k,\Bbb R)$, $\gd sl(k,\Bbb
R)=sl(k,\Bbb R)$ and $\gd o(k)=o(k)$, i.e., all derivations of
these Lie algebras are inner, and the algebras $sl(k,\Bbb R)$ and
$o(k)$ have a zero center.
\end{ex}

\begin{ex} \label{s6} \mar{s6}
Let $X$ be an $n$-dimensional smooth real manifold and $\cT_1(X)$
a real Lie algebra of vector fields on $X$ with respect to the Lie
bracket
\be
u\cdot v=[u,v], \qquad u,v\in \cT_1(X).
\ee
All derivations of a Lie algebra $\cT_1(X)$ are inner and this
algebra has no divisors of zero. Consequently, $\gd \cT_1(X)=
\cT_1(X)$.
\end{ex}

\begin{ex} \label{s6a} \mar{s6a}
Let $Y\to X$ be a smooth fibre bundle and $\cV(Y)$ a real Lie
algebra of vertical vector fields on $Y$. Then $\gd \cV(Y)$ is a
Lie algebra of projectable vector fields on $Y$.
\end{ex}

\begin{ex} \label{s7} \mar{s7} Let $(Z,w)$ be a Poisson
manifold endowed with a Poisson bivector $w$. Let $C^\infty(Z)$ be
a real Poisson algebra of smooth real functions on $Z$ with
respect to a Poisson bracket
\mar{s8}\beq
f\cdot g=\{f,g\}, \qquad f,g\in C^\infty(Z). \label{s8}
\eeq
Its derivations need not be inner. For instance, let $(Z=\Bbb
R^{2m},w)$ be a symplectic manifold which is coordinated by
$(q^i,p_i)$, $i=1,\ldots,m$, and provided with the symplectic
structure
\be
\Om=dp_i\w dq^i, \qquad \{f,g\}=\dr^if\dr_ig-\dr_if\dr^ig, \qquad
f,g\in C^\infty(Z).
\ee
Then $\dr_1:f\to \dr_1f$ is a non-inner derivation of the Poisson
algebra $C^\infty(Z)$ (\ref{s8}).
\end{ex}

\begin{ex} \label{s11} \mar{s11}
Let $V$ be a real vector space, which need not be
finite-dimensional. Let us consider the space $W=V\oplus V\oplus
\Bbb R$. We denote its elements by $w=(\pi,\f,\la)$. Let $<v,v'>$
be a scalar product (a positive non-degenerate bilinear form) on
$V$. Then one brings $W$ into a real Lie algebra with respect to
the bracket
\mar{s9}\beq
w\cdot w'=[w,w']=-[w',w]=(<\pi,\f'>  - <\pi',\f>)\bb, \qquad
\forall w,w'\in W. \label{s9}
\eeq
It is readily observed that the Jacobi identity of this Lie
algebra is trivial, i.e.,
\mar{s10}\beq
w\cdot w'\cdot w''=0, \qquad \forall  w,w',w''\in W. \label{s10}
\eeq
The Lie algebra (\ref{s9}) is called the algebra of canonical
commutation relations (henceforth CCR). Its inner derivations read
\be
\dr_w(\pi',\f',\la')= (0,0,<\pi,\f'> - <\pi',\f>), \quad
w=(\pi,\f,\la)\in W, \quad \forall (\pi',\f',\la')\in W.
\ee
A generic derivation of the CCR algebra $W$ (\ref{s9}) takes the
form
\mar{s99}\beq
\dr w=(M_1(\pi)+O_1(\f),O_2(\pi) +M_2(\f),C_1(\pi)+C_2(\f)),
\label{s99}
\eeq
where $C_1$ and $C_2$ are linear functions on $V$, and $M_1$,
$M_2$, $O_1$ and $O_2$ are endomorphisms of $V$ such that
\be
&& <M_1(v),v'>+ <v,M_2(v')>=0, \qquad \forall v,v'\in V,\\
&& <O_1(v),v'> -<v,O_1(v')>=0, \qquad <O_2(v),v'> -<v,O_2(v')>=0.
\ee
\end{ex}

\section{Differential operators on Lie algebras}

We consider the following class of differential operators on a Lie
$\cK$-algebra $A$.

\begin{defi} \label{s14} \mar{s14} A zero order differential operator $\Phi$ on a Lie algebra
$A$ is defined as an endomorphism of an $(A-A)$-module $A$, i.e.,
\be
\Phi(a\cdot b)=a\cdot\Phi(b)=\Phi(a)\cdot b.
\ee
\end{defi}

It may happen that zero order differential operators on a Lie
algebra $A$ are reduced to multiplications
\mar{s35}\beq
A\ni a\to\la a\in A, \qquad \la\in \cK, \label{s35}
\eeq
i.e. $\dif_0(A)=\cK$. This is the case of finite-dimensional Lie
algebra in Example \ref{s3}, a Lie algebra of vector fields
$\cT_1(X)$ in Example \ref{s6}, a CCR algebra in Example
\ref{s11}, and a Poisson algebra $C^\infty(Z)$ on a symplectic
manifold $Z$ in Example \ref{s7}. If $(Z,w)$ is a non-symplectic
Poisson manifold, zero order differential operators on a Poisson
algebra $C^\infty(Z)$ are not exhausted by the multiplications
(\ref{s35}). Namely, if $c(z)$ are Casimir functions, the
morphisms
\be
C^\infty(Z)\ni f(z)\to c(z)f(z)\in C^\infty(Z)
\ee
are zero order differential operators on a Poisson algebra
$C^\infty(Z)$, i.e., $\dif_0(A)=\cZ[C^\infty(Z)]$.

\begin{defi} \label{s15} \mar{s15} A first order differential operator on a Lie algebra
$A$ is defined as a sum $\dr +\Phi$ of a derivation $\dr$ of a Lie
algebra $A$ and a zero order differential operator $\Phi$ on $A$
(cf. the formula (\ref{s95'})).
\end{defi}

This definition is based on the following facts.

\begin{lem} \label{s24} \mar{s24}
Composition of a zero order differential operator on a Lie algebra
and a derivation of this Lie algebra is a first order differential
operator in accordance with Definition \ref{s15}. The bracket of a
zero order differential operator on a Lie algebra and a derivation
of this Lie algebra is a zero order differential operator.
\end{lem}

\begin{proof} Given a zero order differential operator $\Phi$ on $A$ and a derivation $\dr$ of $A$,
we have the equalities
\mar{s31,2}\ben
&& (\Phi\circ\dr)(a\cdot b)=\Phi(\dr a\cdot b +a\cdot\dr b)= (\Phi(\dr
a))\cdot b + a\cdot \Phi(\dr b), \label{s31}\\
&& (\Phi\circ\dr - \dr\circ\Phi)(a\cdot b)= \Phi(\dr a\cdot b) +\Phi(a\cdot\dr
b) - \dr a\cdot\Phi(b) - a\cdot\dr(\Phi(b))= \label{s32}\\
&& \qquad a(\Phi\circ\dr - \dr\circ\Phi)(b). \nonumber
\een
The equality (\ref{s31}) shows that $\Phi\circ\dr$ is a derivation
of $A$. By virtue of the equality (\ref{s32}), $\Phi\circ\dr -
\dr\circ\Phi$ is a zero order differential operator and,
consequently, $\dr\circ\Phi$ is a first order differential
operator in accordance with Definition \ref{s15}.
\end{proof}

Note that the equality (\ref{s32}) is not trivial if zero order
differential operators and derivations do not commute with each
other. This is the case of a Lie algebra $\cV(Y)$ of vertical
vector fields on a fibre bundle $\pi:Y\to X$ in Example \ref{s6a}.
Zero order differential operators on this Lie algebra are the
morphisms
\be
\cV(Y)\ni u\to fu\in \cV(Y), \qquad \forall f\in \pi^*C^\infty(X),
\ee
i.e., $\dif_0(\cV(Y))=C^\infty(X)$. They do not commute with the
non-inner derivations of $\cV(Y)$ which are non-vertical
projectable vector fields on $Y$.

Definitions \ref{s14} and \ref{s15} of zero and first order
differential operators on Lie algebras are similar to those of
zero and first order differential operators on a commutative ring.
A difference is that all zero order differential operators on a
commutative ring are multiplications in this ring, while
multiplications (\ref{s1}) in a Lie algebra are first order
differential operators.

\begin{defi} \label{s20} \mar{s20} A differential operator of order $k>1$
on a Lie algebra $A$
is defined as a composition of $k$ first order differential
operators on $A$.
\end{defi}

Definition \ref{s20} is not exhausted. There exist different
morphisms of Lie algebras which can be treated as differential
operators as follows.

\begin{ex} \label{s16} \mar{s16}
Let $Y\to X$, coordinated by $(x^\m,y^i)$, be a vector bundle with
a structure group $U(k)$ which is provided with a constant fibre
metric $\eta$. Let $\cV(Y)$ be a Lie algebra of vertical vector
fields on $Y$ in Example \ref{s6a}. An endomorphism
\be
\cV(Y)\ni u^k\dr_k\to \eta^{ij}\dr_i\dr_j(u^k)\dr_k\in\cV(Y).
\ee
of $\cV(Y)$ as a real vector space is a second order differential
operator on a $C^\infty(Y)$-module $\cV(Y)$. This endomorphism
also may be regarded as a second order differential operator on a
Lie algebra $\cV(Y)$, but it fails to satisfy Definition
\ref{s20}.
\end{ex}

Due to Lemma \ref{s24}, it is easily justified that differential
operators on a Lie algebra in accordance with Definitions
\ref{s14}, \ref{s15} and \ref{s20} possess the following
properties.

\begin{lem} \label{s17} \mar{s17} A composition $\Delta\circ\Delta'$
of two differential operators $\Delta\in\dif_k(A)$ and
$\Delta'\in\dif_m(A)$ of order $k$ and $m$, respectively, is a
$(k+m)$-order differential operator.
\end{lem}

\begin{lem} \label{s18} \mar{s18} Given differential operators
$\Delta\in\dif_k(A)$ and $\Delta'\in\dif_m(A)$ of order $0<k$ and
$m\leq k$, respectively, their bracket $[\Delta,\Delta']$ is a
differential operator of order $m$.
\end{lem}

\begin{ex} \label{s42} \mar{s42}
Differential operators on the Lie algebras $gl(k,\Bbb R)$ and
$o(k)$ in Example \ref{s3} are exhausted by compositions of
multiplications in these algebras.
\end{ex}

\begin{ex} \label{s90} \mar{s90} Let $\cT_1(X)$ be a Lie algebra
of vector fields on a manifold $X$ in Example \ref{s6a}. In
accordance with Definition \ref{s15}, a first order differential
operator $\Delta$ on $\cT_1(X)$ takes the form
\mar{s91}\beq
\Delta(v)=[u,v] +\la v, \qquad \forall v\in\cT_1(X), \label{s91}
\eeq
where $u\in \cT_1(X)$ and $\la\in \Bbb R$. Accordingly, a
$k$-order differential operator on $\cT_1(X)$ is a composition of
$k$ first order operators (\ref{s91}). At the same time,
$\cT_1(X)$ is a $C^\infty(X)$-module. By virtue of Definition
\ref{s93}, zero order differential operator on a
$C^\infty(X)$-module $\cT_1(X)$ are its endomorphisms, i.e,
\be
\dif_0(\cT_1(X),\cT_1(X))=\hm_{C^\infty(X)}(\cT_1(X),\cT_1(X)).
\ee
A first order differential operator $\Delta$ on a
$C^\infty(X)$-module $\cT_1(X)$ obeys the condition (\ref{s95})
which reads
\be
\Delta(f v)=u(f)v +f\Delta(v), \qquad \forall v\in\cT_1(X), \qquad
\forall f\in C^\infty(X),
\ee
where $u\in \cT_1(X)$. Namely, we have
\be
\Delta(v)=\nabla_u v, \qquad \Delta(f v)=u(f)v +f\nabla_uv, \qquad
\forall v\in \cT_1(X), \qquad \forall f\in C^\infty(X),
\ee
where $\nabla_u$ is a covariant derivative along a vector field
$u\in \cT_1(X)$ with respect to some linear connection on the
tangent bundle $TX$ (cf. the formula (\ref{s91})).
\end{ex}

\begin{ex} \label{s100} \mar{s100} Since zero order differential
operators on a CCR algebra in Example \ref{s11} are its
multiplications (\ref{s35}), higher order differential operators
on this Lie algebra are exhausted by compositions of derivations
(\ref{s99}).
\end{ex}

\section{Differential operators on modules over Lie algebras}

Let a $\cK$-module $P$ be a left module over a Lie algebra $A$
which acts on $P$ by endomorphisms
\be
A\times P\ni (a,p)\to ap\in P, \qquad (a\cdot b)p=[a,b]p= (a\circ
b- b\circ a)p, \qquad \forall a,b\in A.
\ee
In physical application, one can think of $P$ as being a carrier
space of a representation of $A$. Therefore, we consider
$P$-valued differential operators on $P$.

\begin{defi} \label{s21} \mar{s21}
A zero order differential operator $\Phi$ on a module $P$ is its
$A$-module endomorphism, i.e.,
\be
\Phi(ap)=a\Phi(p), \qquad \forall a\in A, \qquad \forall p\in\Phi.
\ee
\end{defi}

It may happen that zero order differential operators on $P$ are
exhausted by multiplications
\be
P\ni p\to \la p\in P, \qquad \forall \la\in \cK.
\ee
For instance, this is the case of Lie algebras $gl(k,\Bbb R)$ and
$o(k)$ in Example \ref{s3} acting in $\Bbb R^k$.

\begin{defi} \label{s25} \mar{s25} A first order differential operator on an
$A$-module $P$ is defined to be a $\cK$-module endomorphism
$\Delta$ of $P$ which obeys the relation
\mar{s23}\beq
\Delta(ap)=(\dr a)\Phi(p) + a\Delta(p), \qquad \forall a\in A,
\qquad \forall p\in P, \label{s23}
\eeq
where $\dr$ is a derivation of a Lie algebra $A$ and $\Phi$ is a
zero order differential operator on $P$.
\end{defi}

For instance, a multiplication
\mar{s22}\beq
P\ni p\to ap\in P, \qquad a\in A, \label{s22}
\eeq
is a first order differential operator on $P$ which satisfies the
condition (\ref{s23}):
\mar{s40}\beq
b(ap)= [b,a]p + a(bp)=(\dr_ba)p +a(bp), \qquad \forall a,b\in A,
\qquad \forall p\in P. \label{s40}
\eeq

If $P=A$, a first order differential operator on a Lie algebra $A$
also is that on a left $A$-module $A$ in accordance with
Definition \ref{s25}, and {\it vice versa}.

Obviously, compositions of zero order and first order differential
operators on $P$ are first order differential operators on $P$.
However, the bracket of a zero order differential operator and a
first differential operator need not be a zero order differential
operator, and the bracket of two first order differential
operators is note necessarily the first order one.

Definitions \ref{s21} and \ref{s25} of zero and first order
differential operators on modules over a Lie algebra are similar
to those of zero and first order differential operators on modules
over a commutative ring. An essential difference is that, in the
case of a ring, the multiplications (\ref{s22}) are zero order
differential operators on $P$.

\begin{defi} \label{s36} \mar{s36} A differential operator of order $k>1$
on an $A$-module $P$ is defined as a composition of $k$ first
order differential operators on $P$.
\end{defi}

By very definition, a composition of two differential operators of
order $k$ and $m$ is a $(k+m)$-order differential operator.

\begin{ex} \label{s43} \mar{s43}
Differential operators on a module $\Bbb R^k$ over a Lie algebra
$gl(k,\Bbb R)$ in Example \ref{s3} are exhausted by action of
elements of the universal enveloping algebra of $gl(k,\Bbb R)$.
\end{ex}

\begin{ex} \label{s41} \mar{s41} Let $Y\to X$ be a fibre bundle and $\cV(Y)$ the Lie
algebra of vertical vector fields on $Y$ in Example \ref{s6a}. A
ring $C^\infty(Y)$ of smooth real functions on $Y$ is provided
with a structure of a left $\cV(Y)$-module with respect to
morphism
\be
&& \cV(Y)\times C^\infty(Y)\ni (u,f)\to uf=u\rfloor df=u^i\dr_i f\in
C^\infty(Y), \\
&& (u\circ v - v\circ u)f=([u,v])f, \qquad \forall u, v\in \cV(Y), \qquad
\forall f\in C^\infty(Y).
\ee
Zero order differential operators on a $\cV(Y)$-module
$C^\infty(Y)$ are $C^\infty(X)$-module endomorphisms of
$C^\infty(Y)$. First order differential operators are exemplified
by $\Bbb R$-module endomorphisms
\be
C^\infty(Y)\ni f\to u\rfloor df\in C^\infty(Y),
\ee
where $u$ is a projectable vector field on $Y$. Note that zero
order differential operators on a $C^\infty(X)$-ring $C^\infty(Y)$
are exhausted by its multiplications, while first order
differential operators on this ring are the morphisms
\be
C^\infty(Y)\ni \f\to u\rfloor df + \f\in C^\infty(Y), \qquad u\in
\cV(Y), \qquad \f\in C^\infty(Y).
\ee
For instance, let $Y\to X$ be a fibre bundle in Example \ref{s16}.
Then a $C^\infty(X)$-ring $C^\infty(Y)$ admits a second order
differential operator
\be
C^\infty(Y)\ni f\to \eta^{ij}\dr_i\dr_jf\in C^\infty(Y).
\ee
This endomorphism also may be regarded as a second order
differential operator on a $\cV(Y)$-module $C^\infty(Y)$, but it
does not satisfy Definition \ref{s36}.
\end{ex}

\begin{ex} \label{s101} \mar{s101} Let us consider a CCR algebra
$W=\Bbb R\oplus\Bbb R\oplus \Bbb R$ in Example \ref{s11}. Given
its basis $(e_\pi,\e_\f)$, the bracket (\ref{s9}) takes the form
\be
[w,w']= (\pi\f'-\pi'\f)\bb, \qquad w=\pi e_\pi +\f e_\f +\la\bb,
\qquad w'=\pi' e_\pi +\f' e_\f +\la'\bb.
\ee
The generic derivation (\ref{s99}) of this Lie algebra reads
\mar{s104}\beq
\dr(\pi e_\pi +\f e_\f +\la\bb)= (M\pi +O_1\f)e_\pi + (-M\f
+O_2\pi)e_\f +(C_1\pi +C_2\f)\bb, \label{s104}
\eeq
where $M,O_1,O_2,C_1,C_2$ are real numbers. Let us consider a ring
$C^\infty(\Bbb R)$ of smooth real functions $f(x)$ on $\Bbb R$. It
is provided with the structure of a $W$-module with respect to the
action
\mar{s103}\beq
w(f)=(\pi\frac{\dr}{\dr x} +\f x +\la)f, \qquad w=\pi e_\pi +\f
e_\f +\la\bb. \label{s103}
\eeq
The action (\ref{s103}) exemplifies the first order differential
operator (\ref{s22}) on a $W$-module $C^\infty(\Bbb R)$. A generic
first order differential operator on this $W$-module reads
\mar{s105}\beq
\Delta=O_1\frac{\dr^2}{\dr x^2} + (C_2-Mx)\frac{\dr}{\dr x}
-(C_1+\frac12O_2x)x+\la, \label{s105}
\eeq
where $M,O_1,O_2,C_1,C_2,\la$ are real numbers. The differential
operator (\ref{s105}) satisfies the condition
\be
\Delta(w f)=(\dr w)f +w(\Delta(f)),
\ee
where $\dr$ is the derivation (\ref{s104}). Any $k$-order
differential operator on a $W$-module $C^\infty(\Bbb R)$ is a
composition of the first order differential operators
(\ref{s105}).
\end{ex}

\section{Differential calculus on Lie algebras}

As was mentioned above, we follow the notion of the
Chevalley--Eilenberg differential calculus over a ring
\cite{dub01,book05,sard9}.

\begin{defi}
Let $\cK$ be a commutative ring, $\cA$ a commutative $\cK$-ring,
$\cZ[\cA]$ the center of $\cA$, and $\gd\cA$ a Lie $\cK$-algebra
of derivations of $\cA$. Let us consider the Chevalley--Eilenberg
complex of $\cK$-multilinear morphisms of $\gd\cA$ to $\cA$, seen
as a $\gd\cA$-module \cite{fuks,book05,sard9}. Its subcomplex of
$\cZ[\cA]$-multilinear morphisms is a differential graded algebra,
called the Chevalley--Eilenberg differential calculus over $\cA$.
\end{defi}

By analogy with this definition, let us construct the
Chevalley--Eilenberg differential calculus over a Lie algebra as
follows.

\begin{defi} \label{s50} \mar{s50}
Let $\cK$ be a commutative ring, $A$ a Lie $\cK$-algebra, $\cZ[A]$
the center of $A$, and $\gd A$ a Lie $\cK$-algebra of derivations
of $A$. Let us consider the Chevalley--Eilenberg complex $C^*[\gd
A;A]$ of $\cK$-multilinear morphisms of $\gd A$ to $A$, seen as a
$\gd\cA$-module \cite{fuks,book05,sard9}. Its subcomplex $\cO^*A$
of $\cZ[A]$-multilinear morphisms $\cO^*A$ is called the
Chevalley--Eilenberg differential calculus over a Lie algebra $A$.
\end{defi}

For instance, let $A$ have no divisors of zero. i.e., $\cZ[A]$=0,
and let all derivations of $A$ be inner, i.e., $\gd A=A$. Then the
Chevalley--Eilenberg differential calculus $\cO^*A$ over $A$ in
Definition \ref{s50} coincides with the well-known
Chevalley--Eilenberg complex $C^*[A]$ of $A$-valued cochains on
$A$ \cite{fuks,book05,sard9}.

For the sake of simplicity, we restrict our consideration to a Lie
algebra $A$ without devisors of zero. In this case, $A$ is an
invariant subalgebra of $\gd A$, and we denote an action of $\gd
A$ on $A$ as
\mar{s52}\beq
\gd A\times A\ni (\ve,a)\to \ve(a)=\ve\cdot a\in A\subset\gd A.
\label{s52}
\eeq
A $\cK$-multilinear skew-symmetric map
\be
c^k:\op\times^k\gd A\to A
\ee
is called an $A$-valued $k$-cochain on the Lie algebra $\gd A$.
These cochains form a $\gd A$-module $C^k[\gd A;A]$ with respect
to an action
\mar{s60}\beq
\ve: c^k\to \ve\cdot c^k, \qquad [\ve,\ve']=\ve\cdot\ve',\qquad
\forall \ve,\ve'\in\gd A, \label{s60}
\eeq
Let us put $C^0[\gd A;A]=A$. We obtain the cochain complex
\mar{spr997}\beq
0\to A\ar^{\dl^0} C^1[\gd A;A]\ar^{\dl^1} \cdots C^k[\gd A;A]
\ar^{\dl^k} \cdots, \label{spr997}
\eeq
with respect to the Chevalley--Eilenberg coboundary operators
\mar{spr132}\ben
&& \dl^kc^k (\ve_0,\ldots,\ve_k)=\op\sum_{i=0}^k(-1)^i\ve_i\cdot c^k(\ve_0,\ldots,
\wh\ve_i, \ldots, \ve_k)+ \label{spr132}\\
&& \qquad \op\sum_{1\leq i<j\leq k}
(-1)^{i+j}c^k(\ve_i\cdot\ve_j, \ve_0,\ldots, \wh\ve_i, \ldots,
\wh\ve_j,\ldots, \ve_k), \nonumber
\een
where the caret $\,\wh{}\,$ denotes omission \cite{fuks}.

In particular, we have
\mar{spr135,6}\ben
&&(\dl^0a)(\ve_0)=\ve_0\cdot a=-\dr_a(\ve_0), \qquad \forall a\in A=C^0[\gd A;A], \label{spr135} \\
&& (\dl^1 c^1)(\ve_0,\ve_1)=\ve_0\cdot c^1(\ve_1)-\ve_1\cdot c^1(\ve_0) -
c^1(\ve_0\cdot\ve_1). \label{spr136}
\een
A glance at the expression (\ref{spr136}) shows that a one-cocycle
$c^1$ on $\gd A$ obeys the relation
\be
c^1(\ve_0\cdot \ve_1)=c^1(\ve_0)\cdot \ve_1+ \ve_0\cdot c^1(\ve_1)
\ee
and, thus, it is an $A$-valued derivation of the Lie algebra $\gd
A$. Accordingly, any one-coboundary (\ref{spr135}) is an inner
derivation $-\dr_a$, $a\in A\subset \gd A$ of $\gd A$.

\begin{lem} \label{s51} \mar{s51}
Any $A$-valued derivation $\dr$ of $\gd A$ is inner, i.e.,
$\dr=\dr_\ve$, $\ve\in\gd A$.
\end{lem}

\begin{proof} If $\dr$ does not vanish on $A\subset \gd A$, it is
an element of $\gd A$ by very definition of $\gd A$. Let
$\dr(a)=0$ for all $a\in A\subset \gd A$. Then we have
\be
0=\dr(\ve\cdot a)=\dr(\ve)\cdot a +\ve\cdot\dr(a)=\dr(\ve)\cdot a,
\qquad \forall a\in A, \qquad \forall \ve\in\gd A.
\ee
Consequently, $\dr(\ve)$ is a zero derivation of $A$, i.e., it is
a zero element of $\gd A$.
\end{proof}

It follows from Lemma \ref{s51} that there is a monomorphism
$C^0[\gd A;A]\to \gd A$, i.e., any one-cocycle $c^1$ is an element
of $\gd A$ such that $c^1(\ve)=\ve\cdot c^1$. Accordingly, one can
think of the cohomology $H^1(\gd A;A)$ of the complex
(\ref{spr997}) as being the set of non-inner derivations of $A$
whose bracket with any derivation of $A$ are inner derivations of
$A$.

In particular, if $\gd A=A$, then the first cohomology of the
complex (\ref{spr997}) is trivial.

\begin{ex}  \label{s55} \mar{s55}
Let $A$ be an $N$-dimensional real Lie algebra provided with a
basis $\{a_m\}$. Let us assume that it has no divisors of zero and
that all its derivations are inner, i.e., $\gd A=A$. Then its
Chevalley--Eilenberg differential calculus (\ref{spr997})
coincides with the Chevalley--Eilenberg complex
\mar{s70}\beq
0\to A\ar^{\dl^0} C^1[A]\ar^{\dl^1} \cdots C^N[A] \label{s70}
\eeq
of a Lie algebra $A$. In particular, the Chevalley--Eilenberg
coboundary operator $\dl^0$ (\ref{spr135}) takes the form
\mar{s56}\beq
(\dl^0a_m)(a_n)=c^r_{nm}a_r, \label{s56}
\eeq
where $c^r_{nm}$ are structure constants. Since one-cocycles
$c^1\in C^1[A]$ are endomorphisms of a vector space $A$, they can
be represented by elements of the tensor product $A\ot A^*$, where
$A^*$ is a coalgebra of $A$. Let $\{\th^r\}$ be the dual basis for
$A^*$ and, accordingly, $\{a_m\ot\th^r\}$ a basis for the vector
space $B^1[A]$ of one-cocycles. In particular, we have
\mar{s63,2}\ben
&& (\dl^0a_m)= c^r_{nm}a_r\ot\th^n, \label{s63}\\
&& \dl^1 (a_m\ot\th^r)=c^n_{pm}a_n\ot (\th^p\w\th^r) -\frac12
c^r_{pn}a_m\ot (\th^p\w\th^n). \label{s62}
\een
One can think of the formula (\ref{s62}) as the Maurer--Cartan
equation. For instance, let $A=sl(k,\Bbb R)$. Then a glance at the
expressions (\ref{s56}) -- (\ref{s62}) shows that the
Chevalley--Eilenberg differential calculus (\ref{s70}) describes
the matrix geometry of a Lie algebra $sl(k,\Bbb R)$ \cite{dub90}.
\end{ex}

\section{Differential operators on
graded Lie algebras}

As was mentioned above, the notion of differential operators and
the differential calculus on a graded commutative ring is a
straightforward generalization of that of differential operators
and differential calculus on a commutative ring
\cite{bart,fuks,book05,sard9}. A difference lies in a definition
of derivation of a graded commutative ring.

Let $\cK$ be a commutative ring. An associative $\cK$-algebra
$\cA$ is called graded if it is endowed with a grading
automorphism $\g$ such that $\g^2=\id$. A graded algebra seen as a
$\Bbb Z_2$-module falls into the direct sum $\cA=\cA_0\oplus
\cA_1$ of two $\Bbb Z$-modules $\cA_0$ and $\cA_1$ of even and odd
elements such that
\be
\g(a)=(-1)^ia, \qquad a\in\cA_i, \qquad i=0,1.
\ee
One calls $\cA_0$ and $\cA_1$ the even and odd parts of $\cA$,
respectively. Since $\g(aa')=\g(a)\g(a')$, we have
\be
[aa']=([a]+[a']){\rm mod}\,2
\ee
where $a\in \cA_{[a]}$, $a'\in \cA_{[a']}$. If $A$ is a ring, then
$[\bb]=0$. A graded algebra $\cA$ is said to be graded commutative
if
\be
aa'=(-1)^{[a][a']}a'a,
\ee
where $a$ and $a'$ are arbitrary homogeneous elements of $\cA$,
i.e., they are either even or odd.

An endomorphism $\dr$ of a graded $\cK$-module $A$ is called a
graded derivation of $A$ if it obeys the graded Leibniz rule
\mar{ws10}\beq
\dr(ab) = \dr(a)b + (-1)^{[a][\dr]}a\dr(b), \qquad \forall a,b\in
\cA. \label{ws10}
\eeq
Graded derivations constitute a graded Lie algebra over a ring
$\cK$ with respect to the graded bracket
\mar{ws14}\beq
[\dr,\dr']=\dr\circ \dr' - (-1)^{[\dr][\dr']}\dr'\circ \dr.
\label{ws14}
\eeq

Let $\Phi\in\hm_\cK(\cA,\cA)$ be an endomorphism of a graded
$\cK$-module $\cA$. Given an element $a\in \cA$, let us define an
endomorphism
\be
(\dl_a\Phi)(c)= a\Phi(c) - (-1)^{[\Phi][a]}\Phi(ac), \qquad
\forall c\in\cA,
\ee
of a $\cK$-module $\cA$, and its endomorphism
$\dl_{a_0}\circ\cdots\circ\dl_{a_k}\Phi$ for any tuple of elements
$a_0,\ldots,a_k$ of $\cA$. Then a definition of differential
operators on a graded commutative ring $\cA$ is a repetition of
Definition \ref{ws131}.

Considering differential operators on graded Lie algebras, one
meets a problem similar to that for Lie algebras. Namely, their
graded derivations $\dr$ do not obey the condition
$\dl_{a_0}\circ\dl_{a_1}\dr=0$.

Let $\cK$ be a commutative ring. A $\Bbb Z_2$-graded
non-associative $\cK$-algebra $A$ (a Lie superalgebra) is called a
graded Lie algebra if its multiplication obeys the relations
\be
&& a\cdot a'=-(-1)^{[a][a']}a'\cdot a,\\
&& (-1)^{[a][a'']}a\cdot a'\cdot a''
+(-1)^{[a'][a]}a'\cdot a''\cdot a + (-1)^{[a''][a']}a''\cdot
a\cdot a' =0.
\ee
Obviously, the even part $A_0$ of a graded Lie $\cK$-algebra $A$
is a Lie $\cK$-algebra. A graded $\cK$-module $P$ is called an
$A$-module if it is provided with a $\cK$-bilinear map
\be
&& A\times P\ni (a,p)\to a p\in P, \qquad [a
p]=([a]+[p]){\rm mod}\,2,\\
&& (a\cdot a')p=(a\circ a'-(-1)^{[a][a']}a'\circ a)p.
\ee

An endomorphism $\dr$ of a graded $\cK$-module $A$ is said to be a
graded derivation of $A$ if it satisfies the graded Leibniz rule
(\ref{ws10}):
\be
\dr(a\cdot b) = \dr(a)\cdot b + (-1)^{[a][\dr]}a\cdot \dr(b),
\qquad \forall a,b\in \cA.
\ee
Graded derivations of a graded Lie algebra $A$ form a graded Lie
$\cK$-algebra $\gd A$ with respect to the graded bracket
(\ref{ws14}).

With this notion of a graded derivation, one can follow
Definitions \ref{s14} -- \ref{s20} and Definition \ref{s50} in
order to describe differential operators and the differential
calculus on a graded Lie algebra.

\begin{ex}  \label{s120} \mar{s120}
Let $A$ be an $N$-dimensional real graded Lie algebra provided
with a basis $\{a_m\}$. Let us assume that it has no divisors of
zero and that all its derivations are inner, i.e., $\gd A=A$. Then
its graded Chevalley--Eilenberg differential calculus coincides
with the graded Chevalley--Eilenberg complex
\mar{s121}\beq
0\to A\ar^{\dl^0} C^1[A]\ar^{\dl^1} \cdots \label{s121}
\eeq
of a graded Lie algebra $A$ \cite{fuks}. In particular, the graded
Chevalley--Eilenberg coboundary operator $\dl^0$ takes the form
\be
(\dl^0a_m)(a_n)=(-1)^{[a_m]}c^r_{nm}a_r,
\ee
where $c^r_{nm}$ are structure constants. Since one-cocycles
$c^1\in C^1[A]$ are endomorphisms of a graded vector space $A$,
they can be represented by elements of the tensor product $A\ot
A^*$, where $A^*$ is a coalgebra of $A$. Let $\{\th^r\}$ be the
dual basis for $A^*$ and, accordingly, $\{a_m\ot\th^r\}$ a basis
for the graded vector space $B^1[A]$ of one-cocycles. In
particular, we have
\be
&& (\dl^0a_m)= (-1)^{[a_m]}c^r_{nm}a_r\ot\th^n, \\
&& \dl^1 (a_m\ot\th^r)=(-1)^{[a_p]([a_p]+[a_m])}c^n_{pm}a_n\ot (\th^p\w\th^r) -\frac12
c^r_{pn}a_m\ot (\th^p\w\th^n),\\
&& \th^p\w\th^n=-(-1)^{[a_p][a_n]}\th^n\w\th^p.
\ee
Thus, one can think of the graded Chevalley--Eilenberg
differential calculus (\ref{s121}) by analogy with Example
\ref{s55} as describing a matrix geometry of a graded Lie algebra
$A$.
\end{ex}

\begin{ex} Let $Y\to X$ be a vector space, and let $(X,\gA)$ be a
graded manifold whose body is $X$ and whose structure sheaf $\gA$
of graded functions is a sheaf of sections of the exterior bundle
\be
\w Y^*=(X\times \Bbb R) \oplus Y^*\oplus\op\w^2 Y^*\oplus\cdots
\ee
where $Y^*$ is the dual of $Y$ \cite{bart,book05,sard9}. Let
$\gd\gA$ be the sheaf of graded derivations of $\gA$. Its sections
are graded derivations of the graded commutative ring $\w Y^*(X)$
of sections of the exterior bundle $\w Y^*$. They are called
graded vector fields on a graded manifold $(X,\gA)$. Given a
trivialization chart $(U;x^\m,y^a)$ of $Y$ and the corresponding
local basis $(x^\m,c^a)$ for $(X,\gA)$, graded vector fields read
\mar{hn14}\beq
u= u^\m\dr_\m + u^a\frac{\dr}{\dr c^a}, \label{hn14}
\eeq
where $u^\la, u^a$ are local graded functions on $U$. Graded
vector fields (\ref{hn14}) form a graded Lie algebra $\gd[\w
Y^*(X)]$ with respect to the graded Lie bracket
\be
[u,u']=u\circ u' -(-1)^{[u] [u']} u'\circ u.
\ee
In particular,
\be
\frac{\dr}{\dr c^a}\circ\frac{\dr}{\dr c^b} =-\frac{\dr}{\dr
c^b}\circ\frac{\dr}{\dr c^a}, \qquad \dr_\m\circ\frac{\dr}{\dr
c^a}=\frac{\dr}{\dr c^a}\circ \dr_\m.
\ee
All graded derivations of the graded Lie algebra $A=\gd[\w
Y^*(X)]$ are inner. Since $A$ has no devisors of zero, we have
$\gd A=A$.
\end{ex}

\begin{ex} \label{s122} \mar{s122}
Let be $V$ be a real vector space. Let us consider the graded
vector space
\be
W=W_1\oplus W_0=(V\oplus V)\oplus \Bbb R.
\ee
We denote its elements by $w=(\pi,\f,\la)$ where $[\pi]=[\f]=1$.
Let $<v,v'>$ be a scalar product on $V$. Then one brings $W$ into
a real graded Lie algebra with respect to the bracket
\mar{s9a}\beq
w\cdot w'=[w,w']=[w',w]=(<\pi,\f'> + <\pi',\f>)\bb, \qquad \forall
w,w'\in W. \label{s9a}
\eeq
The graded Lie algebra (\ref{s9a}) is called the algebra of
canonical anticommutation relations (henceforth CAR). Its inner
derivations read
\be
\dr_w(\pi',\f',\la')= (0,0,<\pi,\f'> + <\pi',\f>), \quad
w=(\pi,\f,\la)\in W, \quad \forall (\pi',\f',\la')\in W.
\ee
A generic derivation of the CAR algebra (\ref{s9a}) takes the form
\mar{s123}\beq
\dr w=(M_1(\pi)+O_1(\f),O_2(\pi) +M_2(\f),C_1(\pi)+C_2(\f)),
\label{s123}
\eeq
where $C_1$ and $C_2$ are linear functions on $V$, and $M_1$,
$M_2$, $O_1$ and $O_2$ are endomorphisms of $V$ such that
\be
&& <M_1(v),v'>+ <v,M_2(v')>=0, \qquad \forall v,v'\in V,\\
&& <O_1(v),v'> +<v,O_1(v')>=0, \qquad <O_2(v),v'> +<v,O_2(v')>=0.
\ee
Let us consider a CAR algebra $W=\Bbb R\oplus\Bbb R\oplus \Bbb R$.
Given its basis $(e_\pi,\e_\f)$, the bracket (\ref{s9a}) takes the
form
\mar{s130}\beq
[w,w']= (\pi\f'+\pi'\f)\bb, \qquad w=\pi e_\pi +\f e_\f +\la\bb,
\qquad w'=\pi' e_\pi +\f' e_\f +\la'\bb. \label{s130}
\eeq
The generic derivation (\ref{s123}) of this Lie algebra reads
\mar{s126}\beq
\dr(\pi e_\pi +\f e_\f +\la\bb)= M\pi e_\pi -M\f e_\f +(C_1\pi
+C_2\f)\bb, \label{s126}
\eeq
where $M,C_1,C_2$ are real numbers. These derivations constitute a
graded Lie algebra $\gd W$. Let us consider a graded commutative
ring $\La$ generated by an odd element $c$. It is a Grassmann
algebra whose elements take the form
\be
h=h_0 +h_1 c, \qquad h_0, h_1\in \Bbb R.
\ee
A Grassmann algebra $\La$ is provided with a structure of a graded
module over the CAR algebra $W$ (\ref{s130}), which acts on $\La$
by the law
\be
&& w=\pi e_\pi +\f e_\f +\la\bb, \qquad h=h_0 +h_1 c,\\
&& w(h)=(\pi \frac{\dr}{\dr c} +\f c +\la)(h) =(\pi h_1+\la h_0) +
(\f h_0 +\la h_1)c.
\ee
This action exemplifies the first order differential operator on a
$W$-module $C^\infty(\Bbb R)$. A generic first order differential
operator on this $W$-module reads
\mar{s127}\beq
\Delta=(C_2-Mc)\frac{\dr}{\dr c} + C_1c+\la, \label{s127}
\eeq
where $M,C_1,C_2,\la$ are real numbers. The differential operator
(\ref{s127}) satisfies the condition
\be
\Delta(w h)=(\dr w)h +w(\Delta(h)),
\ee
where $\dr$ is the derivation (\ref{s126}). Any $k$-order
differential operator on a $W$-module $\La$ is a composition of
first order differential operators (\ref{s127}).
\end{ex}

\end{document}